\documentclass[12pt]{amsart}
\usepackage{amssymb}
\usepackage[usenames]{color}
\usepackage{comment}
\usepackage{hyperref}
\usepackage{qcircuit}
\usepackage{url}
\newtheorem{theorem}{Theorem}
\newtheorem{corollary}[theorem]{Corollary}
\newtheorem{lemma}[theorem]{Lemma}
\newtheorem{proposition}[theorem]{Proposition}

\theoremstyle{definition}
\newtheorem{definition}[theorem]{Definition}
\newtheorem{proviso}[theorem]{Proviso}

\theoremstyle{remark}

\newtheorem{remark}[theorem]{Remark}

\newenvironment{lquote}
  {\list{}{\leftmargin=1.5em\rightmargin=1em}\item[]}%
  {\endlist}
\newenvironment{q}
  {\begin{lquote}}
  {\end{lquote}}

\newcommand{\Ar}{\medskip\noindent\textbf{A:\ }}

\newcommand{\C}{\mathbb C}

\newcommand{\cnot}{\ensuremath{\text{cNot}}}

\newcommand{\Fun}{\text{Fun}}

\newcommand{\ket}[1]{\ensuremath{|#1\rangle}}

\newcommand{\ox}{\ensuremath{\otimes}}

\renewcommand\phi{\varphi}

\newcommand{\Qn}{\medskip\noindent\textbf{Q:\ }}

\newcommand{\s}{\mathcal S}

\newcommand{\Source}{\text{Source}}
\newcommand{\swap}{\text{Swap}}

\newcommand{\Target}{\text{Target}}

\newcommand{\Val}{\text{Val}}

\newcommand{\X}{\mathcal X}
\newcommand{\x}{\times}
\newcommand\xqed[1]{%
  \leavevmode\unskip\penalty9999 \hbox{}\nobreak\hfill
  \quad\hbox{#1}}
\newcommand\tqed{\xqed{$\triangleleft$}}

\title{Circuit pedantry}
\author{Andreas Blass and Yuri Gurevich}

\begin{document}
\thispagestyle{empty}

\begin{abstract}
Boolean and quantum circuits have commonalities and differences. To formalize the syntactical commonality we introduce syntactic circuits where the gates are black boxes. Syntactic circuits support various semantics. One semantics is provided by Boolean circuits, another by quantum circuits. Quantum semantics is a generalization of Boolean but, because of entanglement, the generalization is not straightforward. We consider only unitary quantum circuits here.
\end{abstract}
\maketitle

\section{Introductory dialog} 
\label{sec:intro}

\noindent\textbf{Q\footnotemark:\ } Tell me please what quantum circuits are exactly.
\footnotetext{Quisani, a former student of the second author}

\medskip\noindent\textbf{A\footnotemark:\ }
Can't you consult a textbook on quantum computations?
\footnotetext{The authors speaking one at a time}

\Qn What textbook? Maybe I am being silly or unlucky, but I looked up many textbooks on quantum computations, including the standard text \cite{NC}. Nobody seems to define quantum circuits carefully. Furthermore, nobody seems to define reversible Boolean circuits carefully, at least in the books that I got hold of \cite{Alrabadi,Devos,Morita,Perumala}.

\Ar Four of our Michigan colleagues wrote a good paper on the synthesis of reversible Boolean circuits \cite{SPMH}. They say that a gate is reversible if ``the (Boolean) function it computes is bijective" and that a reversible Boolean circuit is ``an acyclic combinational logic circuit in which all gates are reversible, and are interconnected without fanout."

\Qn Hmm, I have never heard of combinational logic circuits.

\Ar The terminology seems to be used primarily in electrical engineering. But tell us more about what bothers you.

\Qn Presumably, quantum circuits constitute a straightforward generalization of reversible Boolean circuits. But the generalization cannot be too straightforward. A Boolean circuit, furnished with input, allows you to assign Boolean values to each edge of the circuit. As a result it is crystal clear that different ways to evaluate a given circuit on a given input produce the same output. In the case of a quantum circuit with input, you can assign values to edges but, as far as I can see,  this isn't nearly as useful as in the Boolean case. Everybody seems to consider it obvious that different ways to evaluate a given quantum circuit on a given input produce the same output. This is probably true but it needs a proof.
Besides, what is the Boolean analog for measurements?

\Ar We hear you. It is all about circuit pedantry. But let's forget about measurements for the time being, so that our quantum circuits are unitary in the sense that the transformation performed by any gate is unitary.

\Qn OK, I'll bug you about measurements later. Give me a general plan of circuit pedantry as you see it.

\Ar Boolean circuits and unitary quantum circuits have much in common, especially if you abstract from what transformations are assigned to circuit gates. The graph-theoretical foundations are very similar. In this connection, in \S\ref{sec:syntax}, we introduce \emph{syntactic circuits}. The definition of syntactic circuits simplifies in the case of circuits underlying reversible Boolean circuits and unitary quantum circuits; we call them \emph{balanced circuits}. Syntactic circuits support various semantics.

In \S\ref{sec:bool} we look at the semantics of general Boolean circuits and the simplified semantics of Boolean circuits which are \emph{balanced} in the sense that the underlying syntactic circuits are balanced. Reversible Boolean circuits are special balanced Boolean circuits.

In \S\ref{sec:quantum} we look at the semantics of unitary quantum circuits. We'll verify that different evaluations of such a circuit on a given input indeed produce the same result.

\section{Syntactic circuits} 
\label{sec:syntax}

In this section we define a syntactical notion of circuit. It is a graphical structure involving gates which are treated as black boxes.
The Boolean semantics and quantum semantics of circuits will be defined in subsequent sections.

\subsection{Definition}\mbox{} 

Recall that a \emph{directed multigraph} is a 4-tuple $(V,E,\Source,\Target)$ where $V$ is a set of vertices, also called nodes, $E$ is a set of edges, and Source, Target are functions of type $E\to V$ which assign to each edge its source and target nodes respectively. The multigraph is finite if $V$ and $E$ are finite.

A (nonempty directed) \emph{walk} in a multigraph is a nonempty sequence $e_1, e_2, \dots, e_n$ of edges which connects nodes $x_0, x_1, \dots x_n$ so that $\Source(e_i) = x_{i-1}$ and  $\Target(e_i) = x_i$ for every edge $e_i$. We say that this walk is a walk \emph{from $x_0$ to $x_n$}, and that all the nodes $x_0, x_1, \dots, x_n$ are \emph{involved} in the walk. The walk is a \emph{path} if all the nodes are distinct; it is a cycle if the nodes are distinct except  that $x_n = x_0$. The multigraph is \emph{acyclic} if it has no cycles. In an acyclic multigraph, every walk is a path.

\begin{definition}[Syntactic circuits]
\label{def:circuit}
A \emph{syntactic circuit} or simply a \emph{circuit} is a finite acyclic directed multigraph with no isolated nodes and with some additional structure as follows.
\begin{itemize}
\item The nodes are classified as \emph{input nodes}, \emph{output nodes}, and \emph{gates}.
\item The input
    nodes are linearly ordered, and each of them has at least one outgoing edge and no incoming edge. The edges from input nodes are \emph{input edges}.
\item The output nodes are linearly ordered, and each of them has exactly one incoming edge and no outgoing edge. The edges to output nodes are \emph{output edges}.
\item Every gate has at least one incoming edge and at least one outgoing edge. \tqed
\end{itemize}
\end{definition}

We say that a node $x_1$ of a circuit is \emph{earlier} than another node $x_2$ if there is a path from $x_1$ to $x_2$. This relation is a (strict) partial order because the circuit is acyclic.

The \emph{depth} of a circuit is the maximum number of gates involved in any path.

\begin{q}
\Qn Why do you insist on ordering the input and output nodes?

\Ar These orderings are used in computations. Because of these orderings, the inputs and outputs of Boolean circuits are tuples of Boolean values rather than indexed sets of Boolean values. In a spirit of describing rather than prescribing how to work with circuits, we employ the orderings. \tqed
\end{q}

\subsection{Balanced circuits} 
\label{sec:balance}

\begin{definition}
A gate of a circuit is \emph{balanced} if the number of incoming edges equals the number of outgoing edge. That number is the \emph{arity} of the gate. \tqed
\end{definition}

It is common to represent a balanced $r$-ary gate $G$ by a diagram

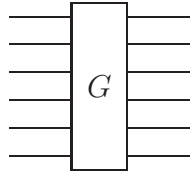
\hspace{4em}
\Qcircuit @C=2em @R=0em {
& \multigate{5}{G} & \qw \\
& \ghost{G} & \qw \\
& \ghost{G} & \qw \\
& \ghost{G} & \qw \\
& \ghost{G} & \qw \\
& \ghost{G} & \qw
}

\medskip\noindent
with $r$ incoming edges on the left and $r$ outgoing edges on the right, so that time flows left to right.

\smallskip
\begin{definition}
A circuit is \emph{balanced} if
\begin{enumerate}
\item all its gates are balanced, and
\item the number of input nodes coincides with the number of output nodes. \tqed
\end{enumerate}
\end{definition}

The number of input nodes is the \emph{width} of a balanced circuit.

The diagram representation of balanced gates naturally extends to balanced  circuits. Here is a simple example \cite{EF} of a width-three circuit

\medskip
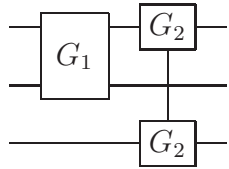
\hspace{4em}
\Qcircuit @C=1em @R=.7em {
& \multigate{1}{G_1} & \sgate{G_2}{2} & \qw \\
& \ghost{G_1} & \qw & \qw\\
& \qw & \gate{G_2} & \qw \\
}\\[1em]
with two binary gates. Let us use this example to illustrate convenient terminology that we will be using.

There are three horizontal strata of edges of the diagram. We follow common terminology calling these strata \emph{timelines}. (This terminology is not intended to mean that some physical entity is propagating along each line.) We number the timelines from top to bottom in diagrams.
Both gates $G_1$ and $G_2$ encounter timeline~1. $G_1$ encounters timeline~2 but $G_2$ doesn't. $G_2$ encounters timeline~3 but $G_1$ doesn't. We say that a gate is \emph{active} on the timelines it encounters. Thus gate $G_2$ is active on timelines 1 and 3.

\begin{proposition}
In any balanced circuit, the input nodes have exactly one outgoing edge each.
\end{proposition}

\begin{proof}
Consider a timeline diagram of a given circuit, and let $w$ be the width of the circuit. The key observation is that the number of timelines remains constant throughout the diagram because all the gates are balanced. At end of the diagram, the timelines are represented by output edges, so the number of timelines is $w$. If at least one input node has two or more outgoing edges then the number of timelines would be $>w$ which is impossible.
\end{proof}

\begin{remark}
A bit more generally, in a circuit where all gates are balanced the number of input edges equals the number of output edges which equals the number of output nodes. Such a circuit is balanced if and only if no input node fans out. \tqed
\end{remark}

\subsection{Composition theory} 
\label{sub:compose}\mbox{}

In the rest of this section, by default, circuits are balanced.

\begin{definition}[Composition]\label{def:compose}
Let $A_1, A_2$ be disjoint circuits of the same width $w$. The \emph{composition} $A_1 A_2$ is the circuit of width $w$ constructed from $A_1, A_2$ thus. At each one of the $w$ strata, merge the output edge of $A_1$ and the input edge of $A_2$ into one edge, removing the output node of $A_1$ and the input node of $A_2$ in the process. \tqed
\end{definition}

If $A_1,A_2$ are not disjoint, replace them with isomorphic copies that are disjoint.
The composite $A_1A_2$ is defined only up to isomorphism. Composition of circuits is associative, so the composite $A_1 A_2 \dots A_N$ is defined for any $N\ge2$. For $N=1$, we adopt the standard convention that the composite of a single circuit is that circuit.

A gate set (i.e.\ a set of gates) $X$ in a circuit is \emph{convex} if every path from an $X$ gate to an $X$ gate involves only $X$ gates.

\begin{definition}[Slice]
Let $X$ be a convex gate set of a circuit $A$.
The \emph{slice of $A$ generated by $X$} is a circuit of the width of $A$ constructed from $A$ as follows.
\begin{enumerate}
\item Remove all non-$X$ gates as well as all edges $e$ such that neither $\Source(e)$ nor $\Target(e)$ belongs to $X$.
\item On every timeline where some $X$ gate is active, do the following. If the input edge has been removed then attach the input node to the unique sourceless edge on the timeline. If the output edge has been removed then attach the output node to the unique targetless edge on the timeline.
\item On every timeline with no active $X$ gates, restore the removed input edge and, if the input edge isn't also the output edge, attach the output node to it. \tqed
\end{enumerate}
\end{definition}

\begin{q}
\Qn Stage (3) looks artificial. Instead of restoring input edges, it seems more natural to create fresh edges.

\Ar It is a bit artificial but for a reason.
Notice that the slice generated by $X$ is completely determined by $X$, that is if $A_1$ is the slice generated by $X$ and so is $A_2$ then $A_1$ and $A_2$ are identical, not almost identical but completely identical.

\Qn Is this important?

\Ar Not really, but it simplifies the exposition, and the price for simplification is small. \tqed
\end{q}

A sequence $X_1, X_2, \dots, X_n$ of gate sets of a circuit $A$ is \emph{coherent} if it respects the relation ``earlier" in the sense that $i<j$ whenever some node of $X_i$ is earlier than some node of $X_j$.
A sequence $X_1, X_2, \dots, X_n$ of gate sets of $A$ is a \emph{coherent partition} of the $A$ gates if it is coherent and if the sets $X_1, X_2, \dots, X_n$ partition the gates of $A$.

\begin{definition}
A \emph{decomposition} of a circuit $A$ is a representation of $A$ as a composite $A_1 A_2 \dots A_n$ where
\begin{enumerate}
\item every $A_i$ is the slice of $A$ generated by a convex subset $X_i$ of the gates of $A$,
\item the sequence $X_1, X_2, \dots, X_n$ is a coherent partition of the gates of $A$. \tqed
\end{enumerate}
\end{definition}

\begin{theorem}\label{thm:decomp}
Let $A$ be a balanced circuit.
\begin{enumerate}
\item Any decomposition $A_1 A_2 \cdots A_n$ of $A$ gives rise to a coherent partition $X_1, X_2, \dots, X_n$ of the $A$ gates where each $X_i$ comprises the gates of $A_i$.
\item Any coherent partition $X_1, X_2, \dots, X_n$ of the $A$ gates gives rise to a decomposition
\[ A = A_1 A_2 \cdots A_n \]
where each factor $A_i$ is the slice generated by $X_i$ (and thus completely determined, not just up to isomorphism).
\end{enumerate}
\end{theorem}

\begin{proof}
The first claim follows directly from the definition of decomposition.

We prove the second claim by induction on $n$. The case $n=1$ is trivial. Suppose that $n>1$. Then $A$ is the composite $A'A_n$ where $A'$ is the slice of $A$ generated by the convex gate set $X_1\cup X_2 \cup \cdots \cup X_{n-1}$. By the induction hypothesis, $A' = A_1 A_2 \cdots A_{n-1}$, and therefore $A = A_1 A_2 \cdots A_n$. The definition of slice implies that the factors are completely determined.
\end{proof}

A subset $X$ of the gates of a circuit is an \emph{antichain} if every two $X$ gates are incomparable, so that neither is earlier than the other. In particular, antichains are convex. Notice that a slice is generated by a nonempty antichain if and only if its depth is 1.

\begin{corollary}\label{cor:sdecomp1}
Let $A$ be a balanced circuit.
\begin{enumerate}
\item Any decomposition $A_1 A_2 \cdots A_T$ of $A$ into depth-1 factors gives rise to a coherent partition of the $A$ gates into antichains $X_1$, $X_2, \dots, X_T$ where each $X_t$ is the gate set of $A_t$.
\item Any coherent partition of the $A$ gates into antichains $X_1, X_2$, \dots, $X_T$ gives rise to a decomposition of $A$ into depth-1 factors
\[ A = A_1 A_2 \cdots A_T \]
where each factor $A_t$ is the slice generated by $X_t$ (and completely determined). \tqed
\end{enumerate}
\end{corollary}

The index $T$ in Corollary~\ref{cor:sdecomp1} alludes to time as will become clear as we discuss semantics in the next two sections.

\begin{definition}
For any circuit $A$, a \emph{coherent gate linearization} is a linear ordering $G_1, G_2, \dots, G_N$ of all the gates which extends the relation ``earlier", so that if $G_i$ is earlier than $G_j$ then $i<j$. \tqed
\end{definition}

\begin{corollary}\label{cor:sdecomp2}
Let $A$ be a balanced circuit with $N$ gates.
\begin{enumerate}
\item Any decomposition $A_1 A_2 \cdots A_N$ of $A$ into one-gate factors gives rise to a coherent gate linearization $G_1, G_2, \dots, G_N$ where each $G_i$ is the only gate of $A_i$.
\item Any coherent gate linearization $G_1, G_2, \dots, G_N$ gives rise to a decomposition into one-gate factors
\[ A = A_1 A_2 \cdots A_N \]
where each factor $A_i$ is the slice generated (and completely determined) by $G_i$. \tqed
\end{enumerate}
\end{corollary}

\section{Boolean circuits} 
\label{sec:bool}

The syntax of Boolean circuits has been defined in the previous section. In this section, we define the semantics.

The central role in Boolean semantics is played by the notion of bit. For computational purposes, a bit is a variable with two possible values $0$ and $1$. Alternatively, a bit can be viewed as a physical system which can be in one of two possible states represented by $0$ and $1$. The first point of view is common in computer science, the second is natural for physics.
For brevity, we call them logical and physical. In this section, the logical point of view will be dominant. The physical point of view is useful as well. In the next section, it will support the generalization of Boolean semantics to quantum semantics.

\subsection{Definition}\mbox{} 

A Boolean function is a function of type $\{0,1\}^m \to \{0,1\}^n$ where $m,n$ are natural numbers. It transforms a state of an $m$-bit system to a state of an $n$-bit system.
For example, the functions
\[ \cnot(a,b) = (a,\ b + a\text{ mod }2)\qquad\text{ and }\qquad
   \swap(a,b) = (b,a) \]
are of type $\{0,1\}^2 \to \{0,1\}^2$. Here are the standard graphical representations

\hspace{4em}
\Qcircuit @C=2em @R=1.5em {
&\ctrl{1} & \qw \\
&\targ & \qw
}
\hspace{1in}
\Qcircuit @C=2em @R=2em {
& \qswap \qwx[1] & \qw \\
& \qswap         & \qw
}

\medskip\noindent
of $\cnot$ and $\swap$ respectively. The function $\cnot$ is known as controlled-not.  The first argument (represented by the black dot) is the control and the second (represented by the circle) is the target.

\begin{definition}\label{def:bcircuit}
  A \emph{Boolean circuit} $B$ is a syntactic circuit $A$ together with an assignment of the following to each gate $G$.
  \begin{enumerate}
  \item A Boolean function $\beta$ of type $\{0,1\}^k \to \{0,1\}^l$ where $k,l$ are the numbers of incoming and outgoing edges of $G$ respectively.
  \item A one-to-one correspondence between the $k$ incoming edges of $G$ and the $k$ argument positions of $\beta$. The edge associated with the $i$-th argument position is the $i$-th \emph{argument edge} of $G$.
  \item A one-to-one correspondence between the $l$ outgoing edges of $G$ and the $l$ value positions of $\beta$. The edge associated with the $j$-th value position is the $j$-th \emph{value edge} of $G$. \tqed
  \end{enumerate}
\end{definition}


\subsection{Valuations} 
\label{val}\mbox{}

Let $B$ be a Boolean circuit with $m$ input nodes.
An \emph{input} to $B$ is an assignment of Boolean values $a_1,\dots,a_m$ to the $m$ input nodes in order. The binary string $a_1\dots a_m$ determines and can be identified with the input.

\begin{definition}
A \emph{valuation} of $B$ on a given input $a_1\dots a_m$  is a Boolean-valued function Val on the edges $e$ of $B$ such that
\begin{itemize}
  \item if $e$ emanates from the $i$-th input node then $\Val(e) = a_i$, and
  \item if $e$ is the $j$-th value edge of a gate $G$ assigned Boolean function $\beta$ and $e_1, e_2, \dots$ are the 1st, 2nd, etc. argument edges of $G$, then
      $\Val(e)\text{ is the $j$-th component of } \beta\big(\Val(e_1),\Val(e_2),\dots\big)$. \tqed
\end{itemize}
\end{definition}

\begin{theorem}
There is a unique valuation of any given Boolean circuit on any given input.
\end{theorem}

\begin{proof}
Construct the desired valuation, in a unique way, by induction on the ordering ``$\Source(e)$ is earlier than $\Source(e')$".
\end{proof}

The \emph{function} of $B$ is the Boolean function
\[ \Fun_B(a_1,\dots,a_m) = (b_1,\dots,b_n)\]
where $b_1,\dots,b_n$ are the Boolean values assigned to the output edges number $1,\dots,n$ respectively by the valuation of $B$ on input $a_1\dots a_m$.

View a Boolean circuit $B$ as a computing device that, given an input binary string $\bar a$, computes the unique valuation on the input and then outputs $\Fun_B(\bar a)$. That computation may be executed by one computing agent or several computing agents which collaborate in performing the circuit computation. In either case, if gate $G_1$ is earlier than gate $G_2$ then $G_1$ must be fired, or executed, before $G_2$ is fired. Incomparable gates may be fired in either order or simultaneously.

We restrict attention to the case where circuit computations are performed by one computing agent executing one step after another.

There may be different such sequential-time computations of a given circuit on a given input.
One of them is the \emph{eager} computation:

{\tt
\hspace{2em} until all gates are fired do\\
\phantom{mmmmmmmmmm}  fire all of the earliest unfired gates.
}

In a general computation, only some of the earliest gates are fired at step~1, only some of the earliest among the remaining gates are fired at step~2, and so on.

\begin{theorem}
  All computations of any given Boolean circuit on any given input produce the same valuation and the same output.
\end{theorem}

\begin{proof}
  Each computation of a given circuit $B$ on a given input $\bar a$ produces a valuation. All these valuations coincide because there is only one valuation of $B$ on $\bar a$. It follows that all the outputs coincide.
\end{proof}

\subsection{Balanced Boolean circuits}\mbox{} 

We turn attention to Boolean circuits whose underlying syntactic circuits are balanced. We simplify the description of such Boolean circuits by imposing a normalizing constraint on the correspondences of incoming and outgoing edges of a gate to the argument and value positions respectively of the Boolean function assigned to the gate.

\begin{definition}\label{def:bbcircuit}
  A \emph{balanced Boolean circuit} $B$ is a balanced syntactic circuit $A$ together with an assignment of the following to each gate $G$.
  \begin{enumerate}
  \item A Boolean function $\beta$ of type $\{0,1\}^r \to \{0,1\}^r$ where $r$ is the arity of $G$.
  \item If $G$ is active on timelines $L_1< L_2< \dots< L_r$ then the incoming and outgoing edges of timeline $L_i$ are assigned to the $i$-th argument position and the $i$-th value position of $\beta$, so that they are the \emph{$i$-th argument edge} and \emph{$i$-th value edge} of $G$ respectively. \tqed
  \end{enumerate}
\end{definition}

The simplification reflects common practice. It is not without a price. One should pay attention to how the argument (resp.\ value) positions of the assigned Boolean function are ordered. For example, one should distinguish between these two versions of the $\cnot$ function:

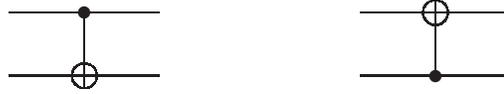
\hspace{4em}
\Qcircuit @C=2em @R=1.5em {
&\ctrl{1} & \qw \\
&\targ & \qw
}
\hspace{1in}
\Qcircuit @C=2em @R=1.5em {
&\targ & \qw \\
&\ctrl{-1} & \qw
}

\begin{remark}
This simplification exploits the linear ordering of timelines on the timeline diagrams that we use to represent balanced syntactic circuits. The incoming (resp.\ outgoing) edges of a gate inherit that order.
The general Boolean circuits can be similarly simplified if appropriate linear orderings of the incoming (resp.\ outgoing) edges of each gate are provided.
\end{remark}

In the rest of this section, Boolean circuits are by default balanced.

\begin{definition}\label{def:bcompose}\mbox{}
The \emph{composition}, also called the \emph{product}, $B_1B_2\dots B_n$ of Boolean circuits $B_1,B_2,\dots,B_n$ with underlying syntactic circuits $A_1,A_2,\dots,A_n$ respectively is a Boolean circuit $B$ whose underlying syntactic circuit is the composition $A_1 A_2 \dots A_n$. On the gates inherited from $A_i$ the gate assignment of $B$ coincides with that of $B_i$. \tqed
\end{definition}

As in the case of syntactic circuits, the composite circuit is defined only up to isomorphism.

The composition of Boolean functions of the same arity is defined as usual:
\begin{align*}
  (\beta_2\beta_1)(\bar a) &= \beta_2(\beta_1(\bar a))\\
  \beta_{n+1}\beta_n\dots\beta_1 &= \beta_{n+1}(\beta_n\dots\beta_1)
\end{align*}

\begin{proposition}
If a Boolean circuit $B$ is the composite of Boolean circuits $B_1, B_2, \dots, B_n$ computing Boolean functions $\beta_1, \dots, \beta_n$ respectively then
\[ \Fun_B = \beta_n \beta_{n-1} \dots \beta_2\beta_1 \]
\end{proposition}

\begin{proof}
It suffices to prove the theorem for $n=2$. But this case of the theorem is a straightforward consequence of Definitions~\ref{def:compose} and \ref{def:bcompose}.
\end{proof}

\begin{definition}
A \emph{decomposition} of a Boolean circuit $B$ is a representation of $B$ as a composite $B_1 B_2 \dots B_n$ such that the composite $A_1 A_2 \dots A_n$ of the underlying syntactic circuits is a decomposition of the underlying syntactic circuit of $B$. \tqed
\end{definition}

Notice that the factors $B_i$ are completely determined by the underlying gate sets.

\begin{definition}
Let $B$ be a balanced Boolean circuit.
\begin{itemize}
\item A Boolean circuit $B'$ is a \emph{slice} of $B$ if the underlying syntactic circuit of $B'$ is a slice of $A$ and the gate assignment of $B$ coincides with that of $B'$ on the $B'$ gates.
\item A \emph{coherent gate linearization} of $B$ is that of  the underlying syntactic circuit of $B$, that is a linear extension of relation ``earlier" on the gates. \tqed
\end{itemize}
\end{definition}

\begin{theorem}\label{thm:bdecomp}
Let $B$ be a balanced Boolean circuit.
\begin{enumerate}
\item Any decomposition $B_1 B_2 \dots B_n$ of $B$ gives rise to a coherent partition $X_1, X_2, \dots, X_n$ of the $B$ gates where each $X_i$ comprises the gates of $B_i$.
\item Any coherent partition $X_1, X_2, \dots, X_n$ of the $B$ gates gives rise to a decomposition
\[ B = B_1 B_2 \dots B_n \]
where each factor $B_i$ is the slice of $B$ generated by $X_i$ (and completely determined). \tqed
\end{enumerate}
\end{theorem}

\begin{proof}
This is a straightforward consequence of Theorem~\ref{thm:decomp}
\end{proof}

\begin{corollary}\label{cor:bdecomp1}
Let $B$ be a balanced Boolean circuit.
\begin{enumerate}
\item Any decomposition $B_1 B_2 \dots B_T$ of $B$ into depth-1 factors gives rise to a coherent partition of the $B$ gates  into antichains $X_1$, $X_2, \dots, X_T$ where each $X_t$ is the gate set of $B_t$.
\item Any coherent partition of the $B$ gates into antichains $X_1$, $X_2$, \dots, $X_T$ gives rise to a decomposition into depth-1 factors
\[ B = B_1 B_2 \dots B_T \]
where each factor $B_t$ is the slice generated by $X_t$ (and completely determined). \tqed
\end{enumerate}
\end{corollary}

\begin{corollary}\label{cor:bdecomp2}
Let $B$ be a balanced Boolean circuit with $N$ gates.
\begin{enumerate}
\item Any decomposition $B_1 B_2 \dots B_N$ of $B$ into one-gate factors gives rise to a coherent gate linearization $G_1, G_2, \dots, G_N$ where each $G_i$ is the only gate of $B_i$.
\item Any coherent gate linearization $G_1, G_2, \dots, G_N$ of the $B$ gates gives rise to a decomposition into one-gate factors
\[ B = B_1 B_2 \dots B_N \]
where each factor $B_i$ is the slice generated by $G_i$ (and completely determined). \tqed
\end{enumerate}
\end{corollary}

Define the depth of a Boolean circuit $B$ to be the depth of the underlying syntactic circuit $A$, that is the maximum number of gates involved in any path.

\begin{theorem}
Let $d$ be the depth of a balanced Boolean circuit $B$, and consider computations of $B$ on a given input.
\begin{enumerate}
\item Any computation has at least $d$ steps.
\item The eager computation has exactly $d$ steps.
\item For some choices of $B$, there are non-eager $d$-step computations.
 \end{enumerate}
\end{theorem}

\begin{proof}\mbox{}
\begin{enumerate}
\item Since a gate cannot be fired until all the earlier gates have been fired, the last gate involved in a longest path can not be fired before step $d$.
\item Obvious.
\item Consider a coherent partition $\dots, X_2, X_1$ the $B$ gates where $X_1$ comprises the latest gates, and $X_2$ comprises the latest among the remaining gates, and so on. It gives rise to a $d$-step computation which may be different from the eager computation. (This computation procrastinates as much as possible subject to finishing in $d$ steps.) \qedhere
\end{enumerate}
\end{proof}

\subsection{Reversible Boolean circuits}\mbox{} 
\label{sec:reverse}

A function $\beta$ of type $\{0,1\}^m \to \{0,1\}^n$ has a (two-sided) inverse if and only if $m=n$ and $\beta$ is a permutation of $\{0,1\}^n$. The inverse $\beta^{-1}$ is unique; it is the inverse of $\beta$ in the permutation group of $\{0,1\}^n$.
For example the controlled-not function  $\cnot(a,b) = (a,b + a\mod 2)$ is its own inverse.

\begin{definition}
A Boolean function is \emph{reversible} if it has a (two-sided) inverse.
A Boolean circuit $B$ is \emph{reversible} if its function $\Fun_B$ is reversible. \tqed
\end{definition}

\begin{lemma}\label{lem:rev}
The composite $B = B_1 B_2 \dots$ of Boolean circuits is reversible if and only if each factor $B_i$ is reversible.
\end{lemma}

\begin{proof}
Let $\delta, \beta_1, \beta_2, \dots$ be the functions of circuits $B, B_1, B_2, \dots$, respectively. We prove that $\delta$ is reversible if and only every $\beta_t$ is reversible.
The ``if\," implication is obvious.

\noindent
The ``only-if\," implication is proved by contrapositive. Suppose that some $\beta_i$ is irreversible, and let $j = \min\{i: \beta_i \text{ is irreversible}\}$, so that $\beta_j(x) = \beta_j(y)$ for some $x\ne y$. Let $\alpha, \gamma$ be the products of functions $\{\beta_i: i<j\}$ and functions $\{\beta_i: i>j\}$, respectively. Our choice of $j$ implies that $\alpha$ is reversible. Let $x' = \alpha^{-1}(x)$ and $y' = \alpha^{-1}(y)$. We have $x'\ne y'$ but
\[ \delta(x') = \gamma(\beta_j(x)) = \gamma(\beta_j(y)) = \delta(y') \]
so that $\delta$ is irreversible.
\end{proof}

\begin{corollary}
A Boolean circuit is reversible if and only if all its gates are reversible.
\end{corollary}

\begin{proof}
Consider a circuit $B$ with $N$ gates and recall that the gates are partially ordered by relation ``earlier". Choose a linear order $G_1, G_2, \dots, G_N$ of the gates which respects relation ``earlier". By Corollary~\ref{cor:bdecomp2}, there is a decomposition $B = B_1 B_2 \dots B_N$ where $B_i$ is the slice of $B$ generated by $G_i$ so that $B_i$ is reversible if and only if $G_i$ is reversible. Now apply Lemma~\ref{lem:rev}.
\end{proof}

\section{Quantum circuits} 
\label{sec:quantum}

The syntax of quantum circuits is that of Boolean circuits and is defined in \S\ref{sec:syntax}. In this section, we define the semantics. However, we impose an important constraint on the circuits under consideration.

\begin{proviso}
We restrict attention to unitary quantum circuits, that is quantum circuits where all gate operators are unitary. Below, quantum circuits are by default unitary.
\end{proviso}

\begin{q}
\Qn In other words, a unitary circuit is a circuit without measurements.

\Ar It is essentially so. But unitary operators can be seen as degenerate single-outcome measurements. See for example the general definition of quantum measurements in \S2.2.3 of \cite{NC}. So, more pedantically, a unitary circuit is a circuit where the only measurements are degenerate ones. \tqed
\end{q}

\subsection{Definition}\mbox{} 

In a sense, quantum circuits generalize balanced Boolean circuits. The role of a bit, as a physical system, is played by a \emph{qubit} which is a quantum system with state space $\C^2$ and a fixed orthonormal basis $\{\ket0, \ket1\}$, known as the \emph{computational basis}. The study of Boolean circuits above, in \S\ref{sec:bool}, readily generalizes to the case where two-state physical systems (bits) are replaced with $d$-state physical systems for any fixed positive integer $d$. Similarly, the study of quantum circuits below readily generalizes to the case where qubits are replaced with qudits, with state space $\C^d$, for any fixed positive integer $d$.

Because of entanglement, the Boolean-to-quantum generalization is not straightforward. In particular, the idea of assigning values to the edges of a given circuit is too simplistic for the quantum case.

Dealing with multi-qubit systems, we will always presume that the qubits are distinguishable. The combined system of qubits $Q_1, Q_2, \dots, Q_n$ will be denoted $Q_1\x Q_2\x \dots\x Q_n$.
The state space of the combined system is
$(\C^2)^{\ox n}$ where $(\C^2)^1 = \C^2$ and $(\C^2)^{k+1} = (\C^2)^k \ox \C^2$.

\begin{definition}
  A \emph{(unitary) quantum circuit} $C$ of width $w$ is a balanced syntactic circuit of width $w$ together with the following assignments.
\begin{enumerate}
\item Inputs $1, 2, \dots, w$ are assigned distinct qubits $Q_1, Q_2, \dots, Q_w$ in order. These $w$ qubits are the \emph{qubits of $C$}, and the combined system $Q = Q_1\x Q_2\x \dots\x Q_w$ is the \emph{physical system} of $C$.
\item Each gate $G$ is assigned a unitary operator $U_G$. If the arity of $G$ is $r$ and $G$ encounters timelines $L_1 < L_2 < \dots < L_r$, then $U_G$ operates on (the state space of) the subsystem $Q_{L_1}\x Q_{L_2}\x \dots\x Q_{L_r}$ of $Q$. \tqed
\end{enumerate}
\end{definition}

The intention is that $C$ describes an evolution of its physical system $Q = Q_1\x Q_2\x \dots\x Q_w$.

The state space $(\C^2)^{\ox w}$ of $Q$ has dimension $2^w$. Its orthonormal \emph{computational basis} comprises vectors
\[ \ket{x} = \ket{x_1}\ox \ket{x_2}\ox \cdots\ox \ket{x_w} \]
where $x=(x_1,x_2,\dots,x_w)$ ranges over the set $\{0,1\}^w$ of binary strings of length $w$.
Every permutation $\beta$ of $\{0,1\}^w$ gives rise to a permutation
$ \ket{x} \mapsto \ket{\beta(x)} $
of the basis vectors of $(\C^2)^{\ox w}$ which, by linearity, extends in a unique way to a unitary operator
\[ \sum_{x\in\{0,1\}^n} c_x \ket{x} \mapsto
   \sum_{x\in\{0,1\}^n} c_x \ket{\beta(x)}. \]
The resulting unitary operator often shares the name of the original permutation. Thus we have unitary operators $\cnot$ and $\swap$ on $(\C^2)^{\ox 2}$.

Of course, unitary operators do not have to merely permute the basis vectors. But, by linearity, any unitary operator is determined by its action on the basis vectors.
Here are two important unitary operators on $\C^2$:
\begin{align*}
H\ket0 &= \frac{\ket0+\ket1}{\sqrt2},&
T\ket0 &= \ket0\\
H\ket1 &= \frac{\ket0-\ket1}{\sqrt2},&
T\ket1 &= e^{i\pi/4} \ket1
\end{align*}
The operator $H$ is known as the Hadamard operator.

\subsection{Combinatorics} 
\label{sub:transpose}\mbox{}

We interrupt our circuit exposition in order to prove an auxiliary combinatorial result (which is probably known but we failed to find it in the literature).

Call a linear order $<$ on a poset (partially ordered set) $\s = (S,\prec)$ \emph{coherent} if $a<b$ whenever $a\prec b$.

A linear order $<$ on a finite set $S$ can be transformed into any other linear order $<'$ on $S$ by adjacent transpositions. In other words, there is a sequence $<_1$, $<_2, \dots, <_k$ of linear orders such that $<_1$ is $<$, and $<_k$ is $<'$, and every $<_{i+1}$ is obtained from $<_i$ by transposing one pair of adjacent elements of $<_i$.
The question arises whether, if $<$ and $<'$ are coherent with a partial order $\prec$, the intermediate orders $<_i$ in the transposition sequence can also be taken to be coherent with $\prec$. The following theorem answers this question affirmatively.

\begin{theorem}\label{thm:comb}
Any coherent linear order on a finite poset $\s = (S,\prec)$ can be transformed into any other coherent linear order on $\s$ by adjacent transpositions with all intermediate orders being coherent.
\end{theorem}

\begin{proof}
Fix a finite poset $(S,\prec)$. We start with an observation that if two elements $u,v$ are ordered differently by two coherent linear orders then $u,v$ are incomparable by $\prec$. Indeed, if $u,v$ were comparable then one of the two linear orders would not be coherent.

Define the distance $D(<,<')$ between two coherent linear orders $<$ and $<'$ to be the number of $(<,<')$ differentiating pairs $u,v$ such that $u<v$ but $v<'u$. We claim that if $D(<,<')\ge1$ then there is a  $(<,<')$ differentiating pair $u,v$ such that $u,v$ are adjacent in ordering $<$. It suffices to prove that if $u,v$ is a $(<,<')$ differentiating pair and $u<w<v$ then either $u,w$ or $w,v$ is a  $(<,<')$ differentiating pair, so that $w<'u$ or $v<'w$. But this is obvious. If $u<'w<'v$ then $u<'v$ which is false.

We prove the theorem by induction on the distance $D(<,<')$ between two given coherent linear orders $<$ and $<'$.
If $D(<,<')=0$, the two orders are identical and there is nothing to prove. Suppose $D(<,<')=d\ge1$.

By the claim above there exist $u<v$ such that $u,v$ are adjacent in $<$ but $v<'u$.
By the observation above, $u,v$ are incomparable by $\prec$. Let $<''$ be the order obtained from $<$ by transposing the adjacent elements $u$ and $v$. $<''$ is coherent because $u,v$ is the only $(<,<'')$ differentiating pair and because $u,v$ are incomparable by $\prec$.

It remains to prove that $<''$ can be transformed into $<'$  by adjacent transpositions with all intermediate linear orders respecting $\prec$. But this follows from the induction hypothesis. Indeed, $D(<'',<')=d-1$ because the $(<'',<')$ differentiating pairs are the same as the $(<,<')$ differentiating pairs, except for $u,v$.
\end{proof}

\subsection{Quantum circuit computations}\mbox{}

Consider a quantum circuit $C$ with physical system $Q$. Assume that initially $Q$ is in state \ket{\Psi_0}. Computationally, \ket{\Psi_0} is an \emph{input} of $C$.

We say that a $C$ gate $G$ is \emph{active at qubit $Q_i$} if $G$ encounters timeline $L_i$. It will be convenient to view the unitary operator $U_G$ operating on the whole system $Q$ and not only on the subsystem formed by the qubits where $G$ is active at; $U_G$ works as the identity operator at every qubit $Q_j$ where $G$ is not active.

If gates $G_1, G_2, \dots, G_n$ of $C$ form an antichain $X$ then they are active at disjoint sets of qubits and therefore the operators $U_{G_1}, U_{G_2}, \dots, U_{G_n}$ commute and can be executed in an arbitrary order or simultaneously; their combined action is a unitary operator which will be denoted $U_X$.

Any computation of $C$ on the given input \ket{\Psi_0} works in sequential time, step after step. At each step an antichain of gates is executed. This determines a coherent partition of the $C$ gates into antichains $X_1, X_2, \dots, X_T$ where $X_t$ comprises the gates executed at step $t$ and $T$ is the number of steps. (This is why we chose the symbols $t$ and $T$, as allusions to ``time''.) Any coherent antichain partition $X_1, X_2, \dots, X_T$ of the $C$ gates uniquely determines a computation of $C$ on the given input where $U_{X_1}, U_{X_2}, \dots$ are executed in order; and so the computation can be identified with the coherent antichain partition.

Any computation $X_1, X_2, \dots, X_T$ of $C$ on the given input \ket{\Psi_0} determines a sequence
\[ \ket{\Psi_0}, \ket{\Psi_1}, \dots, \ket{\Psi_T} \]
of states of $Q$ where every $\ket{\Psi_{t+1}} = U_{X_{t+1}}\ket{\Psi_t}$, and \ket{\Psi_T} is the final state and the \emph{output} of the computation.

\begin{theorem}
For any quantum circuit $C$, any two computations of $C$ on any fixed input \ket{\Psi_0} produce the same output.
\end{theorem}

\begin{proof}
In this proof, a computation means a computation of $C$ on input \ket{\Psi_0}. Two computations will be called equivalent if they produce the same input. A computation will be called linear if a single gate is executed at every step.

First we notice that every computation is equivalent to a linear computation. For the sake of completeness, we provide a proof of this claim. Define the linearity deficit $D(\X)$ of a computation $\X$ to be the number of gates executed at the same step of $\X$ as at least one other gate. It suffices to show that every computation $\X$ with positive $D(\X)$ is equivalent to a computation $\X'$ with $D(\X') < D(\X)$. To this end, let $\X$ be a computation with $D(\X)>0$ and $X$ the first antichain in $\X$ with at least two gates. Further, let $G$ be any of the $X$ gates and $Y = X - \{G\}$. Split step $X$ into two smaller steps: first execute $G$ and then all the $Y$ gates, and let $\X'$ be the resulting computation. Clearly $D(\X') < D(\X)$. Since $X$ is an antichain, the unitary transformations $U_G$ and $U_Y$ commute, and therefore the product $U_Y U_G = U_X$. It follows that $\X'$ is equivalent to $\X$.

It remains to prove that any two linear computations are equivalent. Any linear computation $\X = (G_1, G_2, \dots, G_N)$ constitutes a linear ordering that is coherent in the sense that it respects the relation ``earlier" on the gates. By Theorem~\ref{sub:transpose}, any coherent linearization can be transformed to any other coherent linearization by adjacent transpositions with all intermediate orders being coherent. Accordingly it suffices to prove linear computations $\X$ and $\X'$ are equivalent if $\X'$ is obtained from $\X$ by one adjacent transposition. To prove that, let $\X'$ be obtained from $\X = (G_1, G_2, \dots, G_N)$ by transposing gates $G_i$ and $G_{i+1}$. These two gates are incomparable. Indeed, if $G_i$ were earlier than $G_{i+1}$ then $\X'$ would be incoherent, and if $G_{i+1}$ were earlier than $G_i$ then $\X$ would be incoherent. It follows that the two gates are active on disjoint sets of qubits of $C$. Therefore $U_{G_i}$ and $U_{G_{i+1}}$ commute, $U_{G_i} U_{G_{i+1}} = U_{G_{i+1}} U_{G_i}$, and therefore the two linear computations are equivalent.
\end{proof}

\subsection{Generalizations}\mbox{} 
\label{sub:gen}

\Qn If I understood you correctly, the Boolean semantics of \S\ref{sec:bool} easily generalizes from the standard 2-valued variables (bits) to $d$-valued variables, for any fixed $d$. And quantum semantics easily generalizes from qubits, with state space $\C^2$, to qudits, with state space $\C^d$, for any fixed $d$. I think I understand the generalizations. They require no change of syntactic circuits. By the way, I found some papers that actually use qudits with $d>2$ \cite{BRS,CAB}.

But why keep $d$ fixed? Let $d$ vary.
In a classical (in contrast to quantum) circuit, one input node could be 2-valued while another could be 3-valued. A quantum gate might have a qubit-carrying incoming edge and a qutrit-carrying one. The details seem easy to fix. Am I missing something?

\Ar Reversibility imposes an important constraint. Consider a gate $G$ with $k$ incoming edges of capacities $a_1, a_2, \dots, a_k$ and with $l$ outgoing edges of capacities $b_1, b_2, \dots, b_k$. The constraint is
\begin{equation}\label{eq:1}
a_1 a_2 \cdots a_k = b_1 b_2 \cdots b_l.
\end{equation}

\Qn Explain.

\Ar In the classical case, let $\beta$ be the function assigned to $G$. In order for $\beta$ to be reversible, the domain and codomain of $\beta$ should have exactly the same number of elements. In other words, the number $a_1 a_2 \cdots a_k$ of possible input tuples should be equal to the number $b_1b_2\cdots b_l$ of possible output tuples. Hence equation~\eqref{eq:1}. When all the capacities were just 2, this meant simply that $k=l$, as in our definition of balanced circuits, the underlying syntactic circuits of reversible Boolean circuits. But with general capacities, the equation~\eqref{eq:1} does not reduce to $k=l$.

In the quantum case the situation is similar. Let $U$ be the transformation assigned to $G$. The domain of $U$ is $\C^{a_1}\ox \C^{a_2}\ox \dots \C^{a_k}$, and the codomain of $U$ is $\C^{b_1}\ox \C^{b_2}\ox \dots \C^{b_l}$.
In order for $U$ to be reversible, the dimension $a_1 a_2 \cdots a_k$ of the domain and the dimension $b_1 b_2 \cdots b_l$ of the codomain should be equal. Hence equation~\eqref{eq:1}.

\end{document}